\documentclass[a4paper,UKenglish]{llncs}

\usepackage{diagrams,oldlfont,amssymb,epsf,stmaryrd,color,epsfig}

\newbox\tempa
\newbox\tempb
\newdimen\tempc
\def\mud#1{\hfil $\displaystyle{\mathstrut #1}$\hfil}
\def\rig#1{\hfil $\displaystyle{#1}$}
\def\irulehelp#1#2#3{\setbox\tempa=\hbox{$\displaystyle{\mathstrut #2}$}%
                        \setbox\tempb=\vbox{\halign{##\cr
        \mud{#1}\cr
        \noalign{\vskip\the\lineskip}%
        \noalign{\hrule height 0pt}%
        \rig{\vbox to 0pt{\vss\hbox to 0pt{${\; #3}$\hss}\vss}}\cr
        \noalign{\hrule}%
        \noalign{\vskip\the\lineskip}%

        \mud{\copy\tempa}\cr}}%
                      \tempc=\wd\tempb
                      \advance\tempc by \wd\tempa
                      \divide\tempc by 2 }
\def\irule#1#2#3{{\irulehelp{#1}{#2}{#3}%
                     \hbox to \wd\tempa{\hss \box\tempb \hss}}}

\begin{document}
\title{Complementation: a bridge between finite and infinite proofs}

\author{Gilles Dowek\inst{1} \and Ying Jiang\inst{2}}
\institute{Inria and \'Ecole normale sup\'erieure de Paris-Saclay, 
61, avenue du Pr\'esident Wilson,
94235 Cachan Cedex, France,
{\tt gilles.dowek@ens-cachan.fr}.
\and
State Key Laboratory of Computer Science,
Institute of Software, 
Chinese Academy of Sciences,
100190 Beijing, China,
{\tt jy@ios.ac.cn}.}
\maketitle

\begin{abstract}
When a proposition has no proof in an inference system, it is
sometimes useful to build a counter-proof explaining, step by step,
the reason of this non-provability.  In general, this counter-proof is
a (possibly) infinite co-inductive proof in a different inference system.  
In this paper, we show that, for some decidable inference systems, this 
(possibly) infinite proof has a representation as a finite proof in yet 
another system, equivalent to the previous one.  Conversely, to better
explain the reason of the non-provability, we introduce an efficient
method to transform a finite proof into a (possibly) infinite one.
The method is illustrated with an application to non-reachability proofs
in Alternating pushdown systems.
\end{abstract}

\section{Introduction}

When a proposition has no proof in an inference system, it is
sometimes useful to build a {\em counter-proof} explaining the reason
of this non-provability.

To explain this notion of counter-proof, let us consider an inference
system and a proposition that is not provable in it.  Then, for each
inference rule of this system, allowing to derive this proposition, at
least one premise is not provable.  A way to justify the
non-provability of this proposition is to point at, for each of these
rules, such a non provable premise. To justify the non provability of
these premises, we proceed in the same way, co-inductively building
this way a (possibly) infinite tree explaining, step by step, the
reason of the non-provability of the initial proposition.  This tree
is called a {\em counter-proof} of this proposition.

To express the proofs and the counter-proofs in a single system, we
extend the original inference system ${\cal I}$ into a complete
one ${\cal I}_{\cal J}$, called the {\em complementation} of ${\cal
  I}$, allowing to derive sequents of the form $\vdash A$ when the
proposition $A$ has a proof in ${\cal I}$ and $\not\vdash A$ when it
does not.  Hence, in the system ${\cal I}_{\cal J}$, the proofs of
sequents of the form $\vdash A$ are finite, while those of sequents of
the form $\not\vdash A$ are (possibly) infinite.

In this paper, this general notion of complementation is specialized
to some decidable systems. In \cite{DJ}, we introduced a general
method to prove the decidability of provability 
in an inference system
${\cal I}$, by transforming it into an automaton ${\cal A}$---an
inference systems containing introduction rules
only---preserving provability.  In this automaton, the bottom-up
search for a proof of a proposition $A$ succeeds if the proposition has 
a proof, and fails after a finite time if it does not.  A proof in this
automaton ${\cal A}$ can easily be transformed into a proof in the
original system ${\cal I}$, as all the rules of the automaton ${\cal
A}$ are derivable in ${\cal I}$.

Like the original system ${\cal I}$, the automaton ${\cal A}$ can be
extended into a complete system ${\cal A}_{\cal B}$. 
The system ${\cal B}$ is an automaton as well.  
Thus, in the system ${\cal A}_{\cal B}$, the proofs are
always finite.  We show this way that a sequent of the form $\not\vdash A$
has a (possibly) infinite proof in ${\cal I}_{\cal J}$ if and only if
it has a finite proof in ${\cal A}_{\cal B}$.  In other words, a
(possibly) infinite proof in ${\cal I}_{\cal J}$ always has a finite
representation as a proof in ${\cal A}_{\cal B}$.

However, a proof of a sequent $\not\vdash A$ in ${\cal A}_{\cal B}$ is
not very informative as it can only explain the reason of the
non-provability in the automaton ${\cal A}$, but not in the original
system ${\cal I}$.  The third contribution of this paper is to give an
effective and efficient method to transform a finite proof of
$\not\vdash A$ in ${\cal A}_{\cal B}$ into a more informative
(possibly) infinite proof of this sequent in the system ${\cal
  I}_{\cal J}$.

The inspiration for the present paper is threefold.  The first source
of inspiration is the methods developed to build counter-examples
\cite{Clarke} in model-checking \cite{CGP,CJ}. This possibility to
build counter-example is one of the main advantages of model-checking,
as it allows to find the cause of subtle errors in complex designs.
The second is the notion of negation in logic programming, typically
the notion of negation as failure \cite{Clark}, being used in Prolog
\cite{Roussel} and in artificial intelligence systems, such as Planner
\cite{Hewitt}.  The third is the notions of co-inductive definition
and co-inductive proof \cite{JR,Sangiorgi}, and specially the duality
between inductively and co-inductively definable sets.

The rest of this paper is structured as follows.  After recalling the
notion of proof in Section \ref{Sc:Finite-Infinite-Proof}, we show, in
Section \ref{Sc:Complementation}, how an inference system ${\cal I}$
can be extended into a complete system ${\cal I}_{\cal J}$.  In
Section \ref{Sc:Infinite-to-Finite}, we first recall the notion of
automaton ${\cal A}$ associated to an inference system ${\cal I}$.
Then, we complement this automaton into a system ${\cal A}_{\cal B}$,
and show that proofs are always finite in ${\cal A}_{\cal B}$ and that
${\cal A}_{\cal B}$ is equivalent to ${\cal I}_{\cal J}$.  In Section
\ref{Sc:Finite-to-Infinite}, 
which is the main contribution of the paper,
we design an efficient algorithm which
transforms co-inductively a finite proof in ${\cal A}_{\cal B}$ into a
more informative (possibly) infinite proof in ${\cal I}_{\cal J}$.  In
Section \ref{Sc:Application}, the method is illustrated with
an application to non-reachability certificates in Alternating pushdown
systems \cite{BEM}.

\section{Finite and infinite proofs}
\label{Sc:Finite-Infinite-Proof}

Let ${\cal S}$ be a set whose elements are called {\em propositions}. 

\begin{definition}[Inference rule, inference system, finite in conclusions]

An {\em inference rule} is a partial function from ${\cal S}^n$ to ${\cal S}$,
for some natural number $n$, called the {\em number of premises} of this rule.
An {\em inference system} ${\cal I}$ is a set of inference rules.
It defines a function $F_{\cal I}$ from
${\cal P}({\cal S})$ to ${\cal P}({\cal S})$
$${\small F_{\cal I}(X) =\{f(A_1, ..., A_n)~|~\mbox{$f$ in 
${\cal I}$ and $A_1, ..., A_n$ in $X$}\}}$$

An inference system is {\em finite in conclusions} if, 
for each proposition $B$, there is only a finite number of sequences of
propositions $\langle A^1_1, ..., A^1_{n_1} \rangle$, ..., $\langle
A^p_p, ..., A^p_{n_p} \rangle$ from which $B$ can be derived with a
rule of the system.
\end{definition}

The next definition permits to simplify a rule such as
$${\small \irule{P(x)~~~P(x)}{Q(x)}{}}$$
into the equivalent one
$${\small \irule{P(x)}{Q(x)}{}}$$

\begin{definition}[Simplification]
Let $f$ be an inference rule and 
$D$ its domain.  If 
there exist $i$ and $j$, $i \neq j$, such that, for 
all $\langle x_1, ..., x_n \rangle$ in $D$, $x_i = x_j$, 
then we can replace this rule $f$ with the equivalent rule $f'$ defined by 

\medskip
\noindent
{\small$f'(x_1,
  ..., x_{i-1}, x_i, x_{i+1}, ..., x_{j-1}, x_{j+1}, ..., x_n) =$

\hfill $f(x_1, ..., x_{i-1}, x_i, x_{i+1}, ..., x_{j-1}, x_i, x_{j+1}, ..., x_n)$}
\end{definition}

\begin{definition}[Derivable rule]
\label{saturation}
If $g$ is an inference rule with $n$ premises and $f_1$, ..., $f_n$
are functions such that each $f_i$ is either an inference rule with
$m_i$ premises or the identity function, in which case $m_i = 1$, 
then the rule $h$ defined by
{\small $$h(x^1_1, ..., x^1_{m_1}, ..., x^n_1, ..., x^n_{m_n}) = 
g(f_1(x^1_1, ..., x^1_{m_1}), ..., f_n(x^n_1, ..., x^n_{m_n}))$$}
is the {\em derivable rule obtained by 
composing $g$ and $f_1, ..., f_n$}. 

The domain of this function is the set 
$\langle x^1_1, ..., x^1_{m_1}, ..., x^n_1, ..., x^n_{m_n} \rangle$ such that 
$\langle x^1_1, ..., x^1_{m_1} \rangle$ is in the domain of $f_1$, ..., 
$\langle x^n_1, ..., x^n_{m_n} \rangle$ is in the domain of $f_n$, 
and 
$\langle f_1(x^1_1, ..., x^1_{m_1}),$ $..., f_n(x^n_1, ..., x^n_{m_n}) 
\rangle$ is in the domain of $g$.
\end{definition}

\begin{definition}[Continuous, Co-continuous]
A monotone function $F$ from ${\cal P}({\cal S})$ to ${\cal P}({\cal S})$  is {\em continuous} if, 
for all increasing sequences $X_0, X_1, ...$ in ${\cal P}({\cal S})$, 
$F(\bigcup_n X_n) = \bigcup_n F(X_n)$.
It is {\em co-continuous} if, 
for all decreasing sequences $X_0, X_1, ...$ in ${\cal P}({\cal S})$, 
$F(\bigcap_n X_n) = \bigcap_n F(X_n)$.
\end{definition}

\begin{definition}[Proof]
A {\em (possibly) infinite proof}, also called a {\em co-inductive proof},
in an inference system ${\cal I}$ is a
(possibly) infinite tree labeled with propositions such that, when a node
is labeled with a proposition $B$ and its children are labeled with
propositions $A_1$, ..., $A_n$, then $B$ can be derived with a rule of 
${\cal I}$ from the sequence of propositions $\langle A_1, ..., A_n \rangle$.
A proof is a {\em proof of a proposition} $A$ if its root is labeled with $A$.

A proof that is a finite tree is called a {\em finite proof}.
\end{definition}

\begin{proposition}
Let ${\cal I}$ be an inference system. If ${\cal I}$ is finite in conclusions 
then 
\begin{itemize}
\item 
the function $F_{\cal I}$ is monotone, continuous, and co-continuous,  
\item 
it has a least fixed-point and a greatest fixed-point, 
\item 
an element of ${\cal S}$ is in the least fixed point of $F_{\cal I}$ if and only if
it has an finite proof in ${\cal I}$, and it is in
the greatest fixed point of $F_{\cal I}$ if and only
if it has a (possibly) infinite proof in ${\cal I}$.
\end{itemize}
\end{proposition}

\begin{proof}
See \cite{Sangiorgi}.
\end{proof}

\section{Complementation}
\label{Sc:Complementation}

\begin{definition}[Complement, Complementation]
Let ${\cal I}$ be an inference system 
finite in conclusions.
An inference system ${\cal J}$ is said to be a {\em complement} of
${\cal I}$ if it is also finite in conclusions and for each
proposition $B$, if $\langle A^1_1, ..., A^1_{n_1} \rangle$, ...,
$\langle A^p_1, ..., A^p_{n_p} \rangle$ are all the sequences of
premises from which $B$ can be derived with a rule of ${\cal I}$, then 
the sequences of premises from which $B$ can be derived with a rule of 
${\cal J}$ are all the sequences of the form 
$\langle A^1_{j_1}, ..., A^p_{j_p} \rangle$ for some sequence $j_i$.

Let ${\cal S}'$ be the set containing sequents of the form $\vdash A$
and $\not\vdash A$, for $A$ in ${\cal S}$.
Let ${\cal I}$ be an inference system
and ${\cal J}$ be 
a complement of ${\cal I}$.
The inference system ${\cal I}_{\cal J}$, obtained by taking the rule
mapping $\vdash A_1, ..., \vdash A_n$ to $\vdash f(A_1, ..., A_n)$, 
for each inference rule $f$ in ${\cal I}$, and 
the rule mapping $\not\vdash A_1, ..., \not\vdash A_n$ to 
$\not\vdash f(A_1, ..., A_n)$, for each inference rule $f$ in ${\cal J}$, 
is said to be a {\em complementation} of ${\cal I}$.
\end{definition}

\begin{example}
\label{example1}
Consider the language containing 
a constant $\varepsilon$, 
a unary function symbol $a$, 
and unary predicate symbols $P$, $Q$, $R$, $S$, $T$, $U$ and $V$. 
As usual, the term $a(t)$ is written $a t$. 
Let ${\cal I}$ be the inference system containing the following inference rules
$${\small \begin{array}{lllll}
\irule{U(x)}{Q(a x)}{}
~~~~~~~~~~~~~~
&
\irule{V(x)}{Q(a x)}{}
~~~~~~~~~~~~~~
&
\irule{T(x)}{R(a x)}{}
~~~~~~~~~~~~~~
&
\irule{}{T(x)}{}
~~~~~~~~~~~~~~
&
\irule{Q(x)~R(x)}{P(x)}{}
\\
\\
\irule{S(x)}{P(x)}{}
~~~~~~~~~~~~~~
&
\irule{P(a x)}{Q(x)}{}
\end{array}}$$
Then, the system ${\cal J}$ containing the following rules 
$${\small \begin{array}{lllll}
\irule{Q(x)~~S(x)}{ P(x)}{}
~~~~~~~~~~~~~~
&
\irule{R(x)~~S(x)}{P(x)}{}
~~~~~~~~~~~~~~
&
\irule{P(a)}{Q(\varepsilon)}{}
~~~~~~~~~~~~~~
&
\irule{P(a a x)~~~ U(x)~ ~~V(x)}{ Q(a x)}{}
~~~~~~~~~~~~~~
&
\irule{}{ R(\varepsilon)}{}
\\
\\
\irule{T(x)}{R(a x)}{}
~~~~~~~~~~~~~~
&
\irule{}{S(x)}{}
~~~~~~~~~~~~~~~
&
\irule{}{U(x)}{}
~~~~~~~~~~~~~~~
&
\irule{}{V(x)}{}
\end{array}}$$
is a complement of ${\cal I}$ and 
the system ${\cal I}_{\cal J}$ containing the rules
$${\scriptsize 
\begin{array}{lllll}
\irule{\vdash U(x)}{\vdash Q(a x)}{}
~~~~~~~~~~~~~~
&
\irule{\vdash V(x)}{\vdash Q(a x)}{}
~~~~~~~~~~~~~~
&
\irule{\vdash T(x)}{\vdash R(a x)}{}
~~~~~~~~~~~~~~
&
\irule{}{\vdash T(x)}{}
~~~~~~~~~~~~~~
&
\irule{\vdash Q(x)~\vdash R(x)}{\vdash P(x)}{}
\\
\\
\irule{\vdash S(x)}{\vdash P(x)}{}
~~~~~~~~~~~~~~
&
\irule{\vdash P(a x)}{\vdash Q(x)}{}
\\
\\
\irule{\not\vdash Q(x)~\not\vdash S(x)}{\not\vdash P(x)}{}
~~~~~~~~~~~~~~
&
\irule{\not\vdash R(x)~\not\vdash S(x)}{\not\vdash P(x)}{}
~~~~~~~~~~~~~~
&
\irule{\not\vdash P(a)}{\not\vdash Q(\varepsilon)}{}
~~~~~~~~~~~~~~
&
\irule{\not\vdash P(a a x)~\not\vdash U(x)~\not\vdash V(x)}{\not\vdash Q(a x)}{}
~~~~~~~~~~~~~~
&
\irule{}{\not\vdash R(\varepsilon)}{}
\\
\\
\irule{\not\vdash T(x)}{\not\vdash R(a x)}{}
~~~~~~~~~~~~~~
&
\irule{}{\not\vdash S(x)}{}
~~~~~~~~~~~~~~
&
\irule{}{\not\vdash U(x)}{}
~~~~~~~~~~~~~~
&
\irule{}{\not\vdash V(x)}{}
\end{array}}$$
is a complementation of ${\cal I}$.
\end{example}

\begin{definition}[Conjugate function]
The {\em conjugate} $G$ of a function $F$ 
from ${\cal P}({\cal S})$ to ${\cal P}({\cal S})$ 
is the function from ${\cal P}({\cal S})$ to ${\cal P}({\cal S})$ 
defined by 
{\small $$G(X) = {\cal S} \setminus F({\cal S} \setminus X)$$}
\end{definition} 

\begin{lemma}
\label{complement}
If the function $F$ is continuous, then the function 
$G$ is co-continuous and the complement of the least fixed-point 
of $F$ is the greatest fixed point of $G$:
$${\small {\cal S} \setminus (\bigcup_n F^n(\varnothing)) = \bigcap_n G^n({\cal S})}$$
\end{lemma}

\begin{proof}
From the definition of $G$ and the continuity of $F$, we get that $G$ is 
co-continuous.
Then, by induction on $n$, we 
prove that $G^n({\cal S}) = {\cal S}
\setminus F^n(\varnothing)$. As 
${\cal S} \setminus \bigcup_n F^n(\varnothing)
= \bigcap_n ({\cal S} \setminus F^n(\varnothing))$,
we conclude that
${\cal S} \setminus (\bigcup_n F^n(\varnothing)) = \bigcap_n G^n({\cal S})$.
\end{proof}

\begin{lemma}
\label{complement2}
Let ${\cal I}$ be an inference system that has a complement ${\cal J}$. 
Then, the function $F_{\cal J}$ is the the conjugate of the function 
$F_{\cal I}$, that is, for all $X$ in ${\cal P}({\cal S})$, 
$F_{\cal J}(X) = {\cal S} \setminus F_{\cal I}({\cal S} \setminus X)$. 
\end{lemma}

\begin{proof}
Consider a proposition $B$ and let $\langle A^1_1, ..., A^1_{n_1}
\rangle$, ..., $\langle A^p_1, ..., A^p_{n_p} \rangle$ be the 
$p$ sequences from which $B$ can be derived with a rule of ${\cal I}$.
\begin{itemize}
\item
If $B$ is an element of
$F_{\cal J}(X)$, then it is derivable with a rule of ${\cal J}$ from
the premises $A^1_{j_1}, ..., A^p_{j_p}$ in $X$.  Thus, none of these
propositions is in ${\cal S} \setminus X$.  
Each of the sequences $\langle A^1_1, ..., A^1_{n_1} \rangle$, ..., 
$\langle A^p_1, ..., A^p_{n_p} \rangle$ contains a proposition that is not 
in ${\cal S} \setminus X$, thus $B$ is not derivable 
with a rule of ${\cal I}$ from the propositions of ${\cal S} \setminus X$. 
Thus, $B$ is an element of ${\cal S} \setminus F_{\cal I}({\cal S} \setminus X)$.

\item Conversely, if $B$ is an element of ${\cal S} \setminus F_{\cal I}({\cal S}
\setminus X)$, it is not derivable  with a rule of ${\cal I}$
from the propositions of ${\cal S} \setminus X$. 
Then each of the
sequences $\langle A^1_1, ..., A^1_{n_1} \rangle$, ..., $\langle
A^p_1, ..., A^p_{n_p} \rangle$ contains an element $A^i_{j_i}$ that is
not in ${\cal S} \setminus X$.  Therefore, all the propositions
$A^1_{j_1}, ..., A^p_{j_p}$ are in $X$ and hence $B$ is derivable
with a rule of ${\cal J}$ from the propositions of $X$. 
Thus $B$ is an element of $F_{\cal J}(X)$.
\end{itemize}
\end{proof}

\begin{proposition}
\label{lemmacompleteness}
Let ${\cal I}$ be an inference system that has a complement 
${\cal J}$. 
Then, a proposition has a (possibly) infinite proof in ${\cal J}$ if 
and only if it has no finite proof in ${\cal I}$.
\end{proposition}

\begin{proof}
A proposition $A$ has a (possibly) infinite proof in ${\cal J}$
if and only it is an element of the greatest fixed point of the co-continuous 
function $F_{\cal J}$, if and only if it is an element of the 
greatest fixed point of the
co-continuous conjugate function of 
$F_{\cal I}$ (by Lemma \ref{complement2}), 
if and only if it is not an element of the least fixed point of the function 
$F_{\cal I}$ (by Lemma \ref{complement}), 
if and only if it has no finite proof in ${\cal I}$.
\end{proof}

\begin{theorem}[Completeness]
\label{completeness}
Let ${\cal I}$ be an inference system that has a complement 
${\cal J}$. Then, for all propositions $A$ in ${\cal S}$, either $\vdash A$ 
has a finite proof or $\not\vdash A$ has a (possibly) infinite proof in 
${\cal I}_{\cal J}$. Thus, the system ${\cal I}_{\cal J}$ is complete.
\end{theorem}

\begin{proof}
By Proposition \ref{lemmacompleteness}.
\end{proof}

\begin{example}
\label{example2}
The sequent $\vdash P(a)$ does not have a finite proof in the system
${\cal I}_{\cal J}$ of Example \ref{example1}, 
and the sequent $\not\vdash P(a)$ has an infinite proof
$${\small
\hspace{6cm}\irule{\irule{\irule{\irule{\irule{...}
                                    {\not\vdash P(aaa)}
                                    {}
                                  ~~~~~
                               \irule{}{\not\vdash U(aa)}{}
                                 ~~~~~
                               \irule{}{\not\vdash V(aa)}{}
                             }
                             {\not\vdash Q(aa)}
                             {}
                       ~~~~~~~~~~~~~~~~~~~~~
                       \irule{}
                             {\not\vdash S(aa)}
                             {}
                      }
                      {\not\vdash P(aa)}
                      {}
                  ~~~~~~~~~~~~~~~~~~~~
                 \irule{}{\not\vdash U(a)}{}
                  ~~~~~
                 \irule{}{\not\vdash V(a)}{}
               }
               {\not\vdash Q(a)}
               {}
         ~~~~~~~~~~~~~~~~~~~~~~~~~~
         \irule{}
               {\not\vdash S(a)}
               {}
        }
        {\not\vdash P(a)}
        {}}$$
\end{example}

\section{From infinite to finite proofs}
\label{Sc:Infinite-to-Finite}

\begin{definition}[Introduction rule, Automaton] 
Consider a well-founded order $\prec$ on ${\cal S}$. 
A rule $r$ is said to be an {\em introduction} rule with respect to 
$\prec$,
if 
when a proposition $B$ is 
derivable from premises $A_1$, ..., $A_n$ with the rule $r$, 
we have $A_1 \prec B$, ..., $A_n \prec B$.  A {\em automaton}
is an inference system finite in conclusions, 
containing introduction rules only.
\end{definition}

\begin{proposition}
A complement ${\cal B}$ of an automaton ${\cal A}$ is an automaton as well.  
\end{proposition}

\begin{proof}
By definition of the notion of complement.
\end{proof}

\begin{proposition}
\label{finiteinfinite}
Let ${\cal B}$ be an automaton. Then, the least fixed-point and the 
greatest fixed-point of $F_{\cal B}$ are the same. That is, a
proposition has a (possibly) infinite proof in ${\cal B}$ if and only if it 
has a finite proof in ${\cal B}$.
\end{proposition}

\begin{proof}
As the order $\prec$ is well-founded every branch in a proof is finite.
\end{proof}

\begin{theorem}[Finite and infinite proofs]
\label{main}
Let ${\cal I}$ be a system such that 
\begin{itemize}
\item ${\cal I}$ has a complement ${\cal J}$, 
\item ${\cal I}$ is equivalent to an automaton ${\cal A}$, 
\item the automaton ${\cal A}$ has a complement ${\cal B}$. 
\end{itemize}
Then, a sequent has a proof in 
${\cal I}_{\cal J}$---finite if the sequent 
has the form 
$\vdash A$ and (possibly) infinite if it has the form $\not\vdash A$---if 
and only it has a finite proof in ${\cal A}_{\cal B}$.
\end{theorem}

\begin{proof}
By the hypotheses of the theorem and Proposition \ref{lemmacompleteness}, 
the systems ${\cal I}_{\cal J}$ and 
${\cal A}_{\cal B}$ are equivalent. By Proposition \ref{finiteinfinite},
all proofs are finite in ${\cal A}_{\cal B}$.
\end{proof}

\begin{example}
\label{example3}
The inference system ${\cal I}$, defined in Example \ref{example1},
is equivalent to the automaton ${\cal A}$ containing the rules
$${\scriptsize 
\begin{array}{lllll}
\irule{U(x)~T(x)}{P(ax)}{}
~~~~~~~~~~~~~~
&
\irule{V(x)~T(x)}{P(ax)}{}
~~~~~~~~~~~~~~
&
\irule{U(x)}{Q(a x)}{}
~~~~~~~~~~~~~~
&
\irule{V(x)}{Q(a x)}{}
~~~~~~~~~~~~~~
&
\irule{T(x)}{R(a x)}{}
\\
\irule{}{T(x)}{}
\end{array}}$$
A complement ${\cal B}$ of ${\cal A}$ contains the rules
$${\scriptsize 
\begin{array}{lllll}
\irule{U(x)~V(x)}{ P(a x)}{}
~~~~~~~~~~~~~~
&
\irule{ U(x)~ T(x)}{ P(a x)}{}
~~~~~~~~~~~~~~
&
\irule{ T(x)~ V(x)}{ P(a x)}{}
~~~~~~~~~~~~~~
&
\irule{ T(x)}{ P(a x)}{}
~~~~~~~~~~~~~~
& 
\irule{}{ P(\varepsilon)}{}
\\
\\
\irule{ U(x)~ V(x)}{ Q(a x)}{}
~~~~~~~~~~~~~~
&
\irule{}{ Q(\varepsilon)}{}
~~~~~~~~~~~~~~
&
\irule{ T(x)}{ R(a x)}{}
~~~~~~~~~~~~~~
&
\irule{}{ R(\varepsilon)}{}
~~~~~~~~~~~~~~
&
\irule{}{S(x)}{}
\\
\irule{}{U(x)}{}
~~~~~~~~~~~~~~
&
\irule{}{ V(x)}{}
\end{array}}$$
Thus, 
the system  ${\cal A}_{\cal B}$, that is 
a complementation of ${\cal A}$, contains the rules 
$${\scriptsize 
\begin{array}{lllll}
\irule{\vdash U(x)~\vdash T(x)}{\vdash P(ax)}{}
~~~~~~~~~~~~~~
&
\irule{\vdash V(x)~\vdash T(x)}{\vdash P(ax)}{}
~~~~~~~~~~~~~~
&
\irule{\vdash U(x)}{\vdash Q(a x)}{}
~~~~~~~~~~~~~~
&
\irule{\vdash V(x)}{\vdash Q(a x)}{}
~~~~~~~~~~~~~~
&
\irule{\vdash T(x)}{\vdash R(a x)}{}
\\
\\
\irule{}{\vdash T(x)}{}
\\
\\
\irule{\not\vdash U(x)~\not\vdash V(x)}{\not\vdash P(a x)}{}
~~~~~~~~~~~~~~
&
\irule{\not\vdash U(x)~\not\vdash T(x)}{\not\vdash P(a x)}{}
~~~~~~~~~~~~~~
&
\irule{\not\vdash T(x)~\not\vdash V(x)}{\not\vdash P(a x)}{}
~~~~~~~~~~~~~~
&
\irule{\not\vdash T(x)}{\not\vdash P(a x)}{}
~~~~~~~~~~~~~~
& 
\irule{}{\not\vdash P(\varepsilon)}{}
\\
\\
\irule{\not\vdash U(x)~\not\vdash V(x)}{\not\vdash Q(a x)}{}
~~~~~~~~~~~~~~
&
\irule{}{\not\vdash Q(\varepsilon)}{}
~~~~~~~~~~~~~~
&
\irule{\not\vdash T(x)}{\not\vdash R(a x)}{}
~~~~~~~~~~~~~~
&
\irule{}{\not\vdash R(\varepsilon)}{}
~~~~~~~~~~~~~~
&
\irule{}{\not\vdash S(x)}{}
\\
\\
\irule{}{\not\vdash U(x)}{}
~~~~~~~~~~~~~~
&
\irule{}{\not\vdash V(x)}{}
\end{array}}$$
It is easy to check that the sequent $\not\vdash P(a)$, which has an
infinite proof in ${\cal I}_{\cal J}$, has the finite proof in ${\cal A}_{\cal B}$
$${\small \irule{\irule{}{\not\vdash U(\varepsilon)}{}
         ~~~
         \irule{}{\not\vdash V(\varepsilon)}{}
        }
        {\not\vdash P(a)}
        {}}$$
\end{example}

\section{An algorithm to build infinite proofs}
\label{Sc:Finite-to-Infinite}

It follows from Theorem \ref{main} that, 
if a sequent $\not\vdash A$ has a finite proof in the system ${\cal A}_{\cal B}$,
then it has a (possibly) infinite proof in the system ${\cal I}_{\cal J}$. 
This proof can be effectively constructed, using 
the decidability of ${\cal A}_{\cal B}$. 

For instance, consider the systems of Examples \ref{example1} and
\ref{example3}, the sequent $\not\vdash P(a)$ has a proof in ${\cal
A}_{\cal B}$.  Thus, it must have a (possibly) infinite proof in
${\cal I}_{\cal J}$.  An analysis of the rules of
${\cal I}_{\cal J}$ shows that the last rule 
of this proof is either 
$${\small \irule{\not\vdash Q(a)~\not\vdash S(a)}{\not\vdash P(a)}{}}$$
or 
$${\small \irule{\not\vdash R(a)~\not\vdash S(a)}{\not\vdash P(a)}{}}$$

\noindent
Using the decidability of ${\cal A}_{\cal B}$, we get that the sequents
$\not\vdash Q(a)$ and $\not\vdash S(a)$ have proofs 
in ${\cal A}_{\cal B}$, but the sequent $\not\vdash R(a)$ does not.
So, the last rule of the proof 
in ${\cal I}_{\cal J}$ must be the first one of these two rules. 
We can start the proof with this rule
$${\small \irule{\irule{...}{\not\vdash Q(a)}{}~~~\irule{...}{\not\vdash S(a)}{}}
        {\not\vdash P(a)}
        {}}$$
Then, we consider the last rules used in 
${\cal I}_{\cal J}$ to prove, respectively, the sequent $\not\vdash Q(a)$ and $\not\vdash S(a)$, and
so on.  Co-inductively applying the same procedure to the premises of
the rules of the system ${\cal I}_{\cal J}$, selected by the
decidability of ${\cal A}_{\cal B}$, yields the infinite proof given
in Example \ref{example2}.

This method---co-inductively building a (possibly) infinite proof in 
${\cal I}_{\cal J}$ from a finite one in ${\cal A}_{\cal B}$, using the 
decidability of ${\cal A}_{\cal B}$---is effective but very inefficient, 
as we need to use the full decision algorithm for ${\cal A}_{\cal B}$ at each 
step of the construction of the proof in ${\cal I}_{\cal J}$.

In the remainder of this section, we develop a more efficient
algorithm to build such a proof, in some particular cases.
This algorithm is presented as a constructive proof of the statement:
if a sequent has a proof in ${\cal A}_{\cal B}$, then it has a
co-inductive proof in ${\cal I}_{\cal J}$ (Theorem \ref{theo}).

\subsection{Saturated systems}

As in the previous section, we consider a well-founded order $\prec$ on 
propositions and we classify the rules into introduction rules 
and non-introduction rules. 
For each non-introduction rule, we chose to classify the premises into
{\em major} and {\em non-major premises}, 
in such a way that each rule has at least one major premise, 
and we say that a proof is a 
{\em cut} when it ends with a non-introduction rule whose major premises
are all proved with proofs ending with an introduction rule.

\begin{definition}[Cut]
A {\em cut} is a proof of the form
$${\small 
  \irule{\irule{\irule{\rho^1_1}{A^1_1}{}~~...~~\irule{\rho^1_{k_1}}{A^1_{k_1}}{}}
               {A^1}
               {\mbox{intro}}
         ~~~~~~~~~~~~~~~~~~...~~~~~~~~~~
         \irule{\irule{\rho^m_1}{A^m_1}{}~~...~~\irule{\rho^m_{k_m}}{A^m_{k_m}}{}}
               {A^m}
               {\mbox{intro}}
          ~~~~~~~~~~~~~~~~~
         \irule{\pi_{m+1}}{A^{m+1}}{}
         ~~...~~
         \irule{\pi_n}{A^n}{}
        }
        {A}
        {\mbox{non-intro}}}$$
where $A^1, ..., A^m$ are the major premises of the 
non-intro rule.

A proof {\em contains a cut} if one of its subtrees is a cut. 
\end{definition}

\begin{proposition}[Cut free proofs]
\label{key}
A proof is cut-free if and only if it contains introduction rules only. 
\end{proposition}

\begin{proof}
See \cite{DJ}.
\end{proof}

\begin{definition}[Saturated system]
An inference system is said to be {\em saturated} if each time it contains 
\begin{itemize}
\item
a non-introduction rule $g$ with $n$ premises such that the $m$ 
leftmost premises are major and 
\item
$m$ introduction rules $f_1$, ..., $f_m$ 
\end{itemize}
it also contains
the simplification of the derivable rule obtained by composing 
$g$ with $f_1, ..., f_m$ and $n - m$ times the identity function.
\end{definition}

\noindent
In this section, we only consider inference systems ${\cal I}$ that are
included in and equivalent to a saturated system ${\cal I}'$.

\begin{proposition}[Cut elimination]
\label{cutelimination}
Let ${\cal I}'$ be a saturated system. Each proof in ${\cal I}'$ can 
be reduced to a cut free proof,
that contains introduction rules only.
\end{proposition}

\begin{proof}
We can replace each cut with a derivable rule
of ${\cal I}'$, because ${\cal I}'$ is saturated. 
As the number of rules in the proof decreases, this process terminates.
\end{proof}

\begin{proposition}[Automaton]
Let ${\cal I}'$ be a saturated system. Then, 
the automaton ${\cal A}$ formed
with the introduction rules of ${\cal I}'$ is equivalent to ${\cal I}'$.
\end{proposition}

\begin{proof}
By Propositions \ref{key} and \ref{cutelimination}.

\end{proof}
The relation between the inference systems ${\cal I}$, 
${\cal I}_{\cal J}$, ${\cal I}'$, ${\cal A}$ and 
${\cal A}_{\cal B}$ is depicted in the diagram below.
{\small
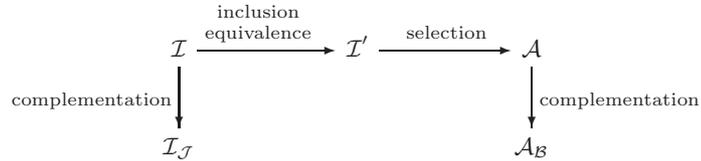
\begin{diagram}[height=2em,width=2.5em]
{\cal I}&\rTo^{\mbox{\hspace{-1cm}{\scriptsize$\begin{array}{c}
\mbox{inclusion}\\
\mbox{equivalence}
\end{array}$}}}&{\cal I}'&\rTo^{\mbox{\scriptsize selection 
\hspace{-0.8cm}}}&{\cal A}\\
\dTo^{\mbox{\scriptsize complementation}}&&&&\dTo_{\mbox{\scriptsize complementation}}\\
{\cal I}_{\cal J}&&&&{\cal A}_{\cal B}\\
\end{diagram}}

\begin{example}
\label{example4}
The system ${\cal I}$ of Example \ref{example1} is included in the 
saturated system ${\cal I}'$ containing the rules of ${\cal I}$ and the 
following 
rules 
$${\small 
\begin{array}{llllll}
\irule{U(x)~T(x)}{P(ax)}{}
~~~~~~~~~~~~~~
& 
\irule{V(x)~T(x)}{P(ax)}{}
~~~~~~~~~~~~~~
&
\irule{U(x)~T(x)}{Q(x)}{}
~~~~~~~~~~~~~~
&
\irule{V(x)~T(x)}{Q(x)}{}
\end{array}}$$
Selecting the introduction rules of ${\cal I}'$
yields the automaton ${\cal A}$ of Example \ref{example3}. 
\end{example}

\subsection{Rank of a rule in a finite saturated system}

\begin{lemma}
\label{aaa}
Consider a non-introduction rule $r$, a derivable rule $r'$ obtained 
from $r$ and introduction rules, a proposition $A$, and a sequence 
$\langle D_1, ..., D_n\rangle$ of premises 
from which $A$ can be derived with $r'$. 
Then, there exists a sequence $\langle C_1, ..., C_p\rangle$ of premises
from which $A$ can be derived with $r$ and 
$\{D_1, ..., D_n\} \prec_{DM} \{C_1, ..., C_p\}$, where 
$\prec_{DM}$ is the multiset extension of $\prec$ \cite{DM}.
\end{lemma}

\begin{proof}
By the construction of $r'$, the multiset $\{D_1, ..., D_n\}$ is
obtained by removing some premises in $\{C_1, ..., C_p\}$ and adding
propositions from which these premises can be derived with
introduction rules, which are strictly smaller
than the removed premises.
\end{proof}

\begin{lemma}
Let ${\cal I}'$ be a finite saturated set of rules and $r_1, r_2, ...,
r_k$ be a sequence of rules of ${\cal I}'$, such that each rule of
this sequence is a derivable rule obtained from the previous one and
introduction rules. Then, the length $k$ of this sequence is lower
than or equal to the number of rules in ${\cal I}'$.
\end{lemma}

\begin{proof}
All we need to prove is that, for all $i$ and $j$ such that $i < j$, one has $r_i \neq r_j$. 
Assume $r_i = r_j$. Consider a proposition $A$ and a sequence $\langle
D_1, ..., D_n\rangle$ of premises from which $A$ can be derived with
$r_j$.  By Lemma \ref{aaa}, there exists a sequence $\langle C_1,
..., C_p\rangle$ of premises from which $A$ can be derived with $r_i$
and $\{D_1, ..., D_n\} \prec_{DM} \{C_1, ..., C_p\} $.  As $r_i =
r_j$, we can build an infinite strictly increasing sequence of set of
premises from which $A$ can be proved with $r_i$, contradicting the
fact that $r_i$ is finite in conclusions.
\end{proof}

\begin{definition}[Rank]
Consider a finite saturated set of rules.  The rank of a rule $r$ in
this set is the length of the longest sequence of rules 
$r_1, r_2, ..., r_k$ starting from $r$, such that each of the rules
$r_2, ..., r_k$ is a non-introduction derivable rule obtained from the 
previous one and introduction rules.
\end{definition}

\subsection{Building the infinite proof}

\begin{proposition}\label{combinatorial}
Consider a natural number $n \geq 1$, $n$ families of sets 
$\langle H^1_1, ..., H^1_{k_1} \rangle$, ..., 
$\langle H^n_1, ..., H^n_{k_n} \rangle$ and 
a set $W$, such that each of the $k_1 \times ... \times 
k_n$ sets of the form 
$H^1_{j_1} \cup ... \cup H^n_{j_n}$ contains an element of $W$.
Then, there exists an index $l$, $1 \leq l \leq n$, such that each of 
the sets $H^l_1, ..., H^l_{k_l}$ contains an element of $W$.
\end{proposition}

\begin{proof}
By induction on $n$. 

If $n = 1$, then each of the sets $H^1_1$, ..., $H^1_{k_1}$ contains an 
element of $W$.

Then, assume the property holds for $n$ and consider 
$\langle H^1_1, ..., H^1_{k_1} \rangle$, ..., 
$\langle H^n_1, ..., H^n_{k_n} \rangle$, 
$\langle H^{n+1}_1, ..., H^{n+1}_{k_{n+1}} \rangle$ 
such that each of the $k_1 \times ... \times k_n \times k_{n+1}$ 
sets of the form 
$H^1_{j_1} \cup ... \cup H^n_{j_n} \cup H^{n+1}_{j_{n+1}}$ contains an element of 
$W$.
We have
\begin{itemize}
\item each of the $k_1 \times ... \times k_n$ sets of the form 
$(H^1_{j_1} \cup ... \cup H^n_{j_n}) \cup H^{n+1}_1$ contains an element of 
$W$, 
\item ..., 
\item each of the $k_1 \times ... \times k_n$ sets of the form 
$(H^1_{j_1} \cup ... \cup H^n_{j_n}) \cup H^{n+1}_{k_{n+1}}$ contains an element of 
$W$.
\end{itemize}
Thus, 
\begin{itemize}
\item either each of the $k_1 \times ... \times k_n$ sets of the form 
$H^1_{j_1} \cup ... \cup H^n_{j_n}$ contains an element of 
$W$ or $H^{n+1}_1$ contains an element of $W$, 
\item ..., 
\item either each of the $k_1 \times ... \times k_n$ sets of the form 
$H^1_{j_1} \cup ... \cup H^n_{j_n}$ contains an element of $W$ 
or $H^{n+1}_{k_{n+1}}$ contains an element of $W$.
\end{itemize}
Hence, 
\begin{itemize}
\item either 
each of the $k_1 \times ... \times k_n$ sets of the form 
$H^1_{j_1} \cup ... \cup H^n_{j_n}$ contains an element of $W$, 
\item or 
$H^{n+1}_1$ contains an element of $W$, ..., and $H^{n+1}_{k_{n+1}}$ contains an 
element of $W$. 
\end{itemize}
Therefore, either, by induction hypothesis, there exists an index $l \leq n$ 
such that each of the sets $H^l_1$, ..., $H^l_{k_l}$ contains an element of $W$,
or each of the sets $H^{n+1}_1$, ..., $H^{n+1}_{k_{n+1}}$ contains an 
element of $W$.
That is, there exists an index $l \leq n+1$ such that 
each of the sets $H^l_1$, ..., $H^l_{k_l}$ contains an element of $W$.
\end{proof}

\begin{proposition}
\label{complement000}
Let ${\cal I}$ be an inference system such that 
\begin{itemize}
\item ${\cal I}$ has a complement ${\cal J}$, 
\item ${\cal I}$ is included in and equivalent to a saturated system 
${\cal I}'$, 
\item and the automaton ${\cal A}$ obtained by 
selecting the introduction rules in ${\cal I}'$ has a complement 
${\cal B}$. 
\end{itemize}
Then, if $A$ is a proposition such that the sequent 
$\not\vdash A$ has a proof in ${\cal A}_{\cal B}$ and 
$\langle C_1, ..., C_n \rangle$ is a sequence of premises
from which $A$ is derivable with a rule of ${\cal I}$, 
then there exists an element $C_i$ in this sequence such that the sequent
$\not\vdash C_i$ has a proof in ${\cal A}_{\cal B}$. 
\end{proposition}

\begin{proof}
By induction on the rank $r$ of the rule that permits to derive $A$
from $\langle C_1, ..., C_n \rangle$ in ${\cal I}$.
\begin{itemize}
\item If $r = 0$, 
the rule permitting to derive $A$ from $\langle C_1, ..., C_n \rangle$ 
is an introduction rule. It is also a rule of ${\cal A}$ and
by construction of ${\cal B}$,  $\langle C_1, ..., C_n \rangle$
contains an element $C_i$ such that the sequent $\not\vdash C_i$
has a proof in ${\cal A}_{\cal B}$. 

\item If $r > 0$, then let $m$ be such that the $m$ leftmost premises
are the major premises of the rule and consider the $k_1$
introduction rules of ${\cal I}'$ with the conclusion $C_1$ and
respective sets of premises $H^1_1$, ..., $H^1_{k_1}$, ..., the
$k_m$ introduction rules of ${\cal I}'$ with the conclusion $C_m$
and respective sets of premises $H^m_1$, ..., $H^m_{k_n}$.  Note
that all these rules are also rules of ${\cal A}$.  Let $k_{m+1} =
... = k_n = 1$ and $H^{m+1}_1 = \{C_{m+1}\}$, ..., $H^n_1 =
\{C_n\}$.  As the system ${\cal I}'$ is saturated it contains $k_1
\times ... \times k_n$ rules with the conclusion $A$ and sets of
premises of the form $H^1_{j_1} \cup ... \cup H^n_{j_n}$.  These
rules have a rank $< r$.  By induction hypothesis each of these $k_1
\times ... \times k_n$ sets contains a proposition $K$ such that
$\not\vdash K$ has a proof in ${\cal A}_{\cal B}$.  Let $W$ be the
finite set of these $k_1 \times ... \times k_n$ propositions.  By
Proposition \ref{combinatorial}, there exists an index $l$ such that
each $H^l_j$ contains a proposition $K^l_j$ in $W$. If $l \leq m$,
then the sequent $\not\vdash K^l_j$ has a proof in ${\cal A}_{\cal
B}$.  Hence, as ${\cal B}$ is a complement of ${\cal A}$,
$\not\vdash C_l$ has a proof in ${\cal A}_{\cal B}$. If $l > m$ then
$K^l_1 = C_l$ and $\not\vdash C_l$ has a proof in ${\cal A}_{\cal
B}$.
\end{itemize}
\end{proof}

\begin{theorem}[Existence of counter-proofs in ${\cal I}_{\cal J}$]
\label{theo}
Let ${\cal I}$ be an inference system such that 
\begin{itemize}
\item ${\cal I}$ has a complement ${\cal J}$, 
\item ${\cal I}$ is included in and equivalent to a saturated system 
${\cal I}'$, 
\item 
the automaton ${\cal A}$ obtained by 
selecting the introduction rules in ${\cal I}'$ has a complement 
${\cal B}$. 
\end{itemize}
Then, if a sequent $\not\vdash A$ has a proof in the system
${\cal A}_{\cal B}$, it has a (possibly) infinite proof 
in the system ${\cal I}_{\cal J}$. 
\end{theorem}

\begin{proof}
By Proposition \ref{complement000}, for each rule of ${\cal I}$
allowing to derive $A$ from $\langle C_1, ..., C_n \rangle$, there
exists an $i$ such that $\not\vdash C_i$ is derivable in ${\cal
A}_{\cal B}$.  As ${\cal J}$ is a complement of ${\cal I}$, there
exists a rule of ${\cal I}_{\cal J}$ deriving $\not\vdash A$ from the
sequents $\not\vdash C_i$.  We co-inductively build a proof of
these sequents.
\end{proof}

\begin{example}
Consider the inference system ${\cal I}_{\cal J}$ and ${\cal A}_{\cal B}$
of Examples \ref{example1} and \ref{example3}. 
We have a proof of $\not\vdash P(a)$ in ${\cal A}_{\cal B}$
$${\small \irule{\irule{}{\not\vdash U(\varepsilon)}{} ~~~
    \irule{}{\not\vdash V(\varepsilon)}{} } {\not\vdash P(a)} {}}$$ To
transform this proof into a (possibly) infinite proof in ${\cal
  I}_{\cal J}$, we use the algorithm described by the proof of Theorem
\ref{theo} step by step. As this proof uses Proposition
\ref{complement000}, whose proof, in turn, uses Proposition
\ref{combinatorial}, this algorithm uses the algorithm described by
the proof of Proposition \ref{complement000}, that, in turn, uses that
described by the proof of Proposition \ref{combinatorial}.  We
describe these three algorithms one after the other.

In the inference system ${\cal I}$ given in Example \ref{example1},
there are two rules to prove $P(a)$.  The first from the sequence of
premises $\langle Q(a), R(a) \rangle$ and the second from the sequence
of premises $\langle S(a) \rangle$.  The proof of Proposition
\ref{complement000} gives us the proposition $Q(a)$ in the first
sequence and the proposition $S(a)$ in the second and proofs of
$\not\vdash Q(a)$ and $\not\vdash S(a)$ in ${\cal A}_{\cal B}$:
$${\small\begin{array}{lll}
\irule{\irule{}{\not\vdash U(\varepsilon)}{}
        ~~~
       \irule{}{\not\vdash V(\varepsilon)}{}}
       {\not\vdash Q(a)} 
        {}
&\hspace{1cm} \mbox{ and } \hspace{1cm}&   
\irule{}
        {\not\vdash S(a)}
        {}   
\end{array}}$$ 
As ${\cal J}$ is a complement of ${\cal I}$, 
there exists a rule in the complementation
${\cal I}_{\cal J}$ that permits to derive $\not\vdash P(a)$ from 
$\not\vdash Q(a)$ and $\not\vdash S(a)$:
$${\small \irule{\not\vdash Q(x)~~~\not\vdash S(x)}{\not\vdash P(x)}{}}$$
So we start the proof with this rule
$${\small \irule{\irule{...}{\not\vdash Q(a)}{}~~~\irule{...}{\not\vdash S(a)}{}}{\not\vdash P(a)}{}}$$
we co-inductively apply the same procedure to the sequents 
$\not\vdash Q(a)$ and $\not\vdash S(a)$, yielding the infinite proof given
in Example \ref{example2}.

Let us focus now on the way the proof of Proposition
\ref{complement000} gives us the proposition $Q(a)$ in the sequence
$\langle Q(a), R(a) \rangle$, and builds the proof of $\not\vdash
Q(a)$ in ${\cal A}_{\cal B}$.

The rule of ${\cal I}$ that permits to derive $P(a)$ from $\langle
Q(a), R(a) \rangle$ has rank $1$ in ${\cal I}'$ of Example \ref{example4}.  
The proposition
$Q(a)$ can be proved in ${\cal I}$ with introduction rules from the
sets $H^1_1 = \{U(\varepsilon)\}$ and $H^1_2 = \{V(\varepsilon)\}$.
The proposition $R(a)$ can be proved in ${\cal I}$ with introduction
rules from the set $H^2_1 = \{T(\varepsilon)\}$.  As the system 
${\cal I}'$ is saturated, it contains two derivable rules allowing to
derive directly $P(a)$ from $H^1_1 \cup H^2_1 = \{U(\varepsilon),
T(\varepsilon)\}$ and from $H^1_2 \cup H^2_1 = \{V(\varepsilon),
T(\varepsilon)\}$.  These rules are
$${\small\begin{array}{lll}
\irule{U(x)~~~T(x)}{P(ax)}{} &\hspace{1cm} \mbox{and} \hspace{1cm}&
\irule{V(x)~~~T(x)}{P(ax)}{}
\end{array}}$$
and both have a rank smaller than $1$, that is $0$.
So they are introduction rules, and hence rules of ${\cal A}$. 
Thus the proposition $P(a)$ can be derived in ${\cal A}$ from the sequence
of premises $\langle U(\varepsilon), T(\varepsilon) \rangle$ and from the 
sequence of premises $\langle V(\varepsilon), T(\varepsilon) \rangle$, respectively. 
From the proof of $\not\vdash P(a)$ in ${\cal A}_{\cal B}$, 
we find a proposition $K$ in each of these sequences, such that 
$\not\vdash K$ has a proof in ${\cal A}_{\cal B}$, that is $U(\varepsilon)$ in the 
first and  $V(\varepsilon)$ in the second, 
together with proofs in ${\cal A}_{\cal B}$
$${\small \begin{array}{lll}
\irule{}{\not\vdash U(\varepsilon)}{} 
&\hspace{1cm} \mbox{ and } \hspace{1cm}&
\irule{}{\not\vdash V(\varepsilon)}{}
\end{array}}$$
Set $W = \{U(\varepsilon), V(\varepsilon)\}$.
The proof of Proposition \ref{combinatorial}, applied to 
$W$,  $H^1_1$,  $H^1_2$ and $H^2_1$, gives us the index $1$.
Thus, both sets $H^1_1$ and $H^1_2$ contains an element 
of $W$: 
$U(\varepsilon)$ and $V(\varepsilon)$ respectively. 
Since these propositions are in $W$, we have proofs of the sequents
$\not\vdash U(\varepsilon)$ and $\not\vdash V(\varepsilon)$ in 
${\cal A}_{\cal B}$.
By the construction of ${\cal B}$, the system ${\cal A}_{\cal B}$ contains a 
rule that permits to derive $\not\vdash Q(a)$ from 
$\not\vdash U(\varepsilon)$ and 
$\not\vdash V(\varepsilon)$. We build this way the proof in ${\cal A}_{\cal B}$ 
$${\small \irule{\irule{}{\not\vdash U(\varepsilon)}{}
         ~~~
         \irule{}{\not\vdash V(\varepsilon)}{}}
        {\not\vdash Q(a)}
        {}}$$
        
Finally, let us focus on the way the proof of
Proposition \ref{combinatorial} gives the index $1$ from the sets
$H^1_1$, $H^1_2$, $H^2_1$, and $W$.  Note that the set $H^1_1 \cup
H^2_1$ is $\{U(\varepsilon)\} \cup \{T(\varepsilon)\}$ and $H^1_2 \cup
H^2_1$ is $\{V(\varepsilon)\} \cup \{T(\varepsilon)\}$.  Then either
each of the sets $\{U(\varepsilon)\}$ and $\{V(\varepsilon)\}$
contains an element of $W$ or $\{T(\varepsilon)\}$ does.  In this
case, each of the sets $\{U(\varepsilon)\}$ and $\{V(\varepsilon)\}$
contains an element of $W$.
So we obtain $l = 1$. 
\end{example}

\section{Application}
\label{Sc:Application}

In this section, we illustrate the main results of the paper 
with an application to Alternating pushdown systems.

\begin{figure}[t]\label{aps}
{\scriptsize
\noindent\framebox{\parbox{\textwidth
}{
{
~~~~~~~~~~~~~~$\begin{array}{ll}
\\
\irule{P_1(x)~...~P_n(x)}
      {Q(a x)}
      {\mbox{intro}~~~n \geq 0}
~~~~~~~~~~~~~~~~~~~~~~~~~~~~~~~~~~~~~~~~~~~~~~~~
&
\irule{P_1(a x)~P_2(x)~...~P_n(x)}
      {Q(x)}
      {\mbox{elim}~~~n \geq 1}
\\
\\
\irule{}
      {Q(x)}
      {\mbox{arbitrary}}
&
\irule{P_1(x)~...~P_n(x)} 
      {Q(x)} 
      {\mbox{neutral}~~~n \geq 1}
\\
\\
\irule{}
      {Q(\varepsilon)}
      {\mbox{empty}}
\end{array}$
\caption{Alternating pushdown systems \label{pushdown}}}}}}
\end{figure}

\begin{definition}[Alternating Pushdown Systems]
\label{APS}
Consider a language ${\cal L}$ containing a finite number of unary
predicate symbols, a finite number of unary function symbols, and a constant $\varepsilon$.  
An {\em Alternating pushdown system} is an
inference system whose rules are like those presented in Figure
\ref{pushdown}, where all premises in a rule are distinct.  
\end{definition}

The rules in the left column of Figure \ref{pushdown}---the intro, 
arbitrary, and empty rules---are introduction rules,
and those in the right column---the elimination and neutral rules---are not. 
Elimination rules have one major premise, the leftmost one, and
all the premises of a neutral rule are major.  

The system ${\cal I}$ introduced in Example \ref{example1} is 
an Alternating pushdown system.

In order to apply Theorems \ref{completeness}, \ref{main}, and
\ref{theo} to Alternating pushdown systems, we need to prove that

\begin{itemize}
\item
for each Alternating pushdown system ${\cal I}$, there exists 
a finite saturated system ${\cal I}'$ such that ${\cal I}$ 
is included in and equivalent to 
${\cal I}'$,

\item
each Alternating pushdown system ${\cal I}$ has a complementation ${\cal I}_{\cal J}$.
\end{itemize}

To build the saturated system ${\cal I}'$, we introduce 
the following saturation procedure that always terminates \cite{DJ}. 

\begin{definition}[Saturation]
\label{saturation}
Each time we have a non-introduction rule $g$ with $n$ premises such
that the $m$ leftmost premises are major and $m$ introduction rules
$f_1$, ..., $f_m$, we add the simplification of the derivable rule
obtained by composing $g$ with $f_1, ..., f_m$ and $n - m$ times the
identity function.
\end{definition}

To prove that each Alternating pushdown system ${\cal I}$ has a 
complementation ${\cal I}_{\cal J}$, we start with the following lemma.

\begin{lemma}
\label{concl}
For each Alternating pushdown system ${\cal I}$, 
there exists an equivalent inference system
$\hat{\cal I}$ and a finite set ${\cal C}$ such that 
\begin{itemize}
\item the conclusions of the rules of $\hat{\cal I}$ are in ${\cal C}$, 

\item for every closed proposition $A$ there exists a unique proposition 
$B$ in ${\cal C}$ such that $A$ is an instance of $B$.
\end{itemize}
\end{lemma}

\begin{proof}
We take for ${\cal C}$, for instance, the set containing all the atomic
propositions of the form $P(\varepsilon)$ and $P(a x)$. 
Then, we replace each neutral rule and elimination rule with the conclusion 
$P(x)$ with an instance with the conclusion $P(\varepsilon)$ and, for each 
function symbol $a$, an instance with the conclusion $P(a x)$.
\end{proof}

\begin{example}\label{hat}
Let ${\cal I}$ be the system introduced in Example \ref{example1}.
The system $\hat{I}$ is 
$${\scriptsize
\begin{array}{lllll}
\irule{U(x)}{Q(a x)}{}
~~~~~~~~~~~~~~
&
\irule{V(x)}{Q(a x)}{}
~~~~~~~~~~~~~~
&
\irule{T(x)}{R(a x)}{}
~~~~~~~~~~~~~~
&
\irule{}{T(\varepsilon)}{}
~~~~~~~~~~~~~~
&
\irule{}{T(a x)}{}
\\
\\
\irule{Q(\varepsilon)~R(\varepsilon)}{P(\varepsilon)}{}
~~~~~~~~~~~~~~
&
\irule{Q(a x)~R(a x)}{P(a x)}{}
~~~~~~~~~~~~~~
&
\irule{S(\varepsilon)}{P(\varepsilon)}{}
~~~~~~~~~~~~~~
&
\irule{S(a x)}{P(a x)}{}
~~~~~~~~~~~~~~
&
\irule{P(a)}{Q(\varepsilon)}{}
\\
\\
\irule{P(a a x)}{Q(a x)}{}
\end{array}}$$
\end{example}

\begin{definition}[Complement]\label{complement_APD}
Let ${\cal I}$ be an Alternating pushdown system, 
$\hat{\cal I}$ the system built at Lemma \ref{concl}, 
and ${\cal C}$ be a finite set of atomic propositions such that 
\begin{itemize}
\item the conclusions of the rules of $\hat{\cal I}$ are in the set 
${\cal C}$, 

\item for every closed proposition $A$, there exists a unique proposition 
$B$ in ${\cal C}$ such that $A$ is an instance of $B$.
\end{itemize}
Then, we define the system ${\cal J}$ as follows: 
for each $B$ in ${\cal C}$, if the system $\hat{\cal I}$ contains
$n$ rules $r^B_1, ..., r^B_n$ with the conclusion $B$, 
$$\irule{A^1_1~...~A^1_{m_1}}
        {B}
        {}$$
$$...$$
$$\irule{A^n_1~...~A^n_{m_n}}
        {B}
        {}$$
where $n$ may be zero
and each $m_i$ ($1 \leq i \leq n$) may be zero as well,
then the system ${\cal J}$ contains the $m_1 \times ... \times 
m_n$ rules 
$$\irule{A^1_{j_1}~...~A^n_{j_n}}
        {B}
        {}$$
\end{definition}

\begin{proposition}
\label{complement200}
The system ${\cal J}$ is a complement of  $\hat{\cal I}$: 
for each proposition $B$, 
if $\langle A^1_1, ..., A^1_{n_1} \rangle$, ...,
$\langle A^p_1, ..., A^p_{n_p} \rangle$ are all the sequences of 
premises from 
which $B$ is derivable in $\hat{\cal I}$, 
then the 
sequences of 
premises from which $B$ is derivable in 
${\cal J}$ are all the sequences of the form 
$\langle A^1_{j_1}, ..., A^p_{j_p} \rangle$, for some sequence $j_i$.
\end{proposition}

\begin{proof}
Consider a closed proposition $B$.
There exists a 
unique proposition $C$ in ${\cal C}$ such that 
$B = \sigma C$. Consider all the rules with conclusion $C$ and 
let 
$\langle D^1_1, ..., D^1_{n_1} \rangle$, ...,
$\langle D^p_1, ..., D^p_{n_p} \rangle$ be all the premises of these rules. 
The 
sequences of 
premises from 
which $B$ is derivable in $\hat{\cal I}$, 
are 
$\langle \sigma D^1_1, ..., \sigma D^1_{n_1} \rangle$, ...,
$\langle \sigma D^p_1, ..., \sigma D^p_{n_p} \rangle$. 
By construction, ${\cal J}$ contains all the rules of the form 
$$\irule{D^1_{j_1}~...~D^n_{j_n}}
        {C}
        {}$$
Thus, the sequences from which $B$ is derivable in ${\cal J}$ are 
all the sequences 
of the form
$\langle \sigma D^1_{j_1}, ..., \sigma D^p_{j_p} \rangle$, 
for some sequence $j_i$.
\end{proof}

\begin{example}
Consider the system $\hat{\cal I}$ in Example \ref{hat}.
The complement ${\cal J}$ built 
as in Definition \ref{complement_APD} 
contains the following inference rules
$${\scriptsize 
\begin{array}{llllll}
\irule{Q(\varepsilon)~~S(\varepsilon)}{P(\varepsilon)}{}
~~~~~~~~~~~~~~
&
\irule{R(\varepsilon)~~S(\varepsilon)}{P(\varepsilon)}{}
~~~~~~~~~~~~~~
&
\irule{Q(a x)~~S(a x)}{P(a x)}{}
~~~~~~~~~~~~~~
&
\irule{R(a x )~~S(a x)}{P(a x)}{}
~~~~~~~~~~~~~~
&
\irule{P(a)}{Q(\varepsilon)}{}
\\
\\
\irule{P(a a x)~~U(x)~~V(x)}{Q(a x)}{}
~~~~~~~~~~~~~~
&
\irule{}{R(\varepsilon)}{}
~~~~~~~~~~~~~~
&
\irule{T(x)}{R(a x)}{}
~~~~~~~~~~~~~~
&
\irule{}{S(x)}{}
~~~~~~~~~~~~~~
&
\irule{}{U(x)}{}
\\
\\
\irule{}{V(x)}{}
\end{array}}$$
Note that this system is equivalent to 
the complement presented in Example \ref{example1} with slight differences 
due to the systematic definition of the set ${\cal C}$.
\end{example}

\section{Conclusion}\label{Sc:Conclusion}
We have introduced a generic method to construct a counter-proof, when
a proposition fails to have a proof in an inference system ${\cal I}$.
First, we have shown that if the inference system ${\cal I}$ has a
complement ${\cal J}$, then it can be extended to a complete system
${\cal I}_{\cal J}$, allowing to prove sequents of the form $\vdash A$
and $\not\vdash A$ with (possibly) infinite proofs (Theorem
\ref{completeness}).  Then, we have shown that if the inference system
${\cal I}$ is equivalent to an automaton ${\cal A}$, and this
automaton also has a complement ${\cal B}$, then a (possibly) infinite
proof in ${\cal I}_{\cal J}$ always has a finite representation as a
proof in ${\cal A}_{\cal B}$ (Theorem \ref{main}).  Finally, we have
introduced an effective and efficient method to transform a finite
proof of a sequent $\not\vdash A$ in the system ${\cal A}_{\cal B}$
into a (possibly) infinite proof of this sequent in the system ${\cal
  I}_{\cal J}$, explaining, step by step, the reason of the
non-provability (Theorem \ref{theo}).  These results have been
illustrated with an application to Alternating pushdown systems, where
they yield finite non-reachability certificates. They can also be
applied smoothly to other inference systems, such that Dynamic
Networks of Pushdown Systems \cite{BMT,LMW}.

\section*{Acknowledgement}

The authors want to thank Ahmed Bouajjani for enlightening discussions.
This work is supported by the ANR-NSFC project LOCALI (NSFC 61161130530
and ANR 11 IS02 002 01) and the 
Chinese National Basic Research Program (973) Grant No. 2014CB340302.

\end{document}